\documentclass[conference]{IEEEtran}
\IEEEoverridecommandlockouts
\usepackage{cite}
\usepackage{amsmath,amssymb,amsfonts}
\usepackage{algorithmic}
\usepackage{graphicx}
\usepackage{textcomp}
\usepackage{xcolor}
\usepackage{subfigure}
\usepackage{caption}
\newtheorem{lemma}{\text{Lemma}}
\newtheorem{proof}{\text{Proof}}
\usepackage[letterpaper, left=0.626in,right=0.626in,bottom=0.986in,top=0.68in]{geometry}
\def\BibTeX{{\rm B\kern-.05em{\sc i\kern-.025em b}\kern-.08em
    T\kern-.1667em\lower.7ex\hbox{E}\kern-.125emX}}
\begin{document}

\title{RSRP Measurement Based Channel Autocorrelation Estimation for IRS-Aided Wideband Communication \\
\thanks{This work has been submitted to the IEEE for possible publication. Copyright may be transferred without notice, after which this version may no longer be accessible.}
}

\author{\IEEEauthorblockN{He Sun\IEEEauthorrefmark{1},
                     Lipeng Zhu\IEEEauthorrefmark{1},
                     Weidong Mei\IEEEauthorrefmark{4},
                     and Rui Zhang\IEEEauthorrefmark{2}\IEEEauthorrefmark{3}\IEEEauthorrefmark{1}}
\IEEEauthorblockA{\IEEEauthorrefmark{1}Department of Electrical and Computer Engineering, National University of Singapore, Singapore 117583.}
\IEEEauthorblockA{\IEEEauthorrefmark{2}School of Science and Engineering, Shenzhen Research Institute of Big Data.}
\IEEEauthorblockA{\IEEEauthorrefmark{3}Chinese University of Hong Kong, Shenzhen, Guangdong, China 518172.}
\IEEEauthorblockA{\IEEEauthorrefmark{4}National Key Laboratory of Wireless Communications, \\
University of Electronic Science and Technology of China, China 611731. \\
E-mail: {sunele@nus.edu.sg, zhulp@nus.edu.sg, wmei@uestc.edu.cn, elezhang@nus.edu.sg}}                    }

\vspace{-19.6pt}

\maketitle

\begin{abstract}
The passive and frequency-flat reflection of IRS, as well as the high-dimensional IRS-reflected channels, have posed significant challenges for efficient IRS channel estimation, especially in wideband communication systems with significant multi-path channel delay spread. To address these challenges, we propose a novel neural network (NN)-empowered framework for IRS channel autocorrelation matrix estimation in wideband orthogonal frequency division multiplexing (OFDM) systems. This framework relies only on the easily accessible reference signal received power (RSRP) measurements at users in existing wideband communication systems, without requiring additional
pilot transmission. Based on the estimates of channel autocorrelation matrix, the passive reflection of IRS is optimized to maximize the average user received signal-to-noise ratio (SNR) over all subcarriers in the OFDM system. Numerical results verify that the proposed algorithm significantly outperforms existing power-measurement-based IRS reflection designs in wideband channels.
\end{abstract}


\section{Introduction}
Intelligent reflecting surfaces (IRSs) have developed as a promising technology for creating cost-effective and adaptable wireless communication environments \cite{wu2023intelligent}. An IRS typically consists of a planar array of numerous quasi-passive, low-cost reflective elements that can tune the phase shifts and amplitudes of incoming signals.
To effectively harness the potential benefits of IRS for channel reconfiguration, IRS passive reflection/beamforming should be properly designed\cite{wu2023intelligent,mei}. Most of the existing IRS passive reflection designs require channel state information (CSI) or dedicated pilot training. For example, the IRS channels are estimated explicitly or implicitly by exploiting the received complex-valued pilots for IRS reflection optimization\cite{wu2023intelligent,mei,Zhengbx2020,YWJSAC}. However, the pilot signals in existing communication systems (e.g., cellular or WiFi) are designed only for estimating the base station (BS)-user direct channels, with no additional resources allocated for IRS-related channel estimation. As a result, pilot-based IRS channel estimation and reflection pattern optimization would require substantial modifications to current communication protocols, posing significant challenges for practical implementation.

\begin{figure}[t]
  \centering   {\includegraphics[width=0.4136\textwidth]{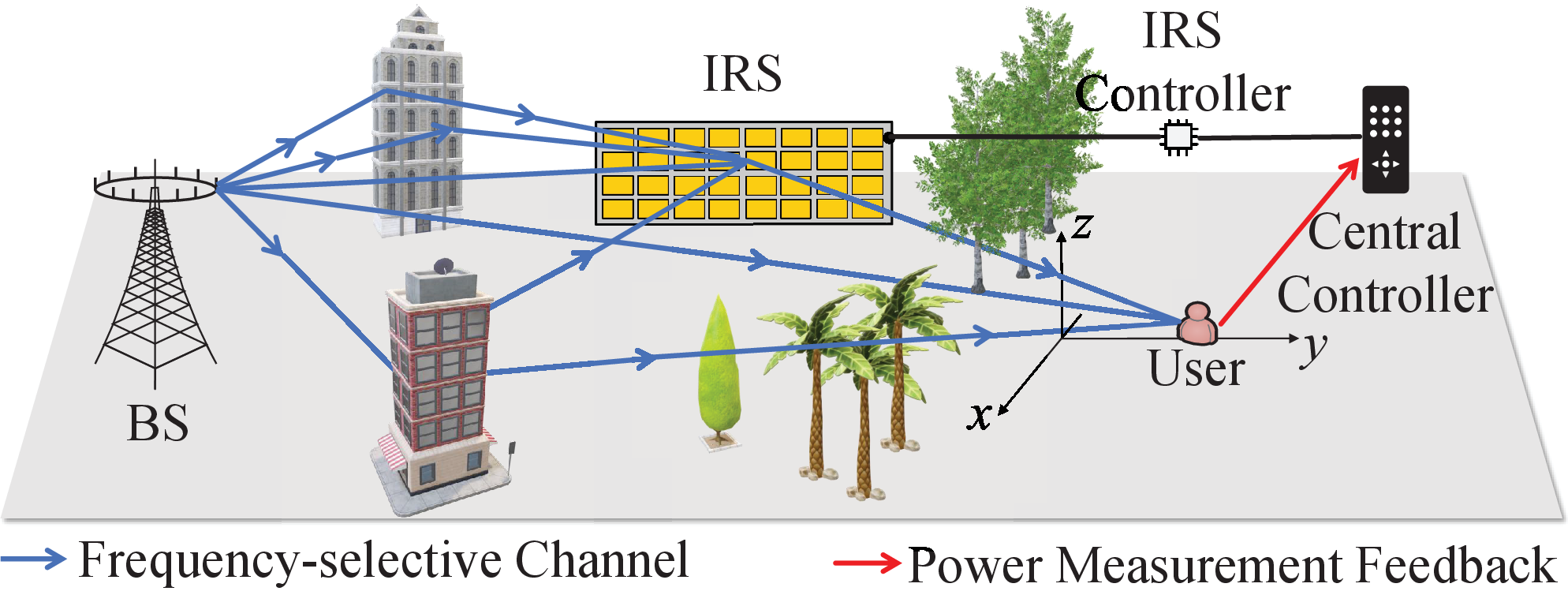}}
  \caption{An IRS-assisted wideband OFDM system. }
\label{fig0001}
\vspace{-16pt}
\end{figure}

To enable seamless integration of IRS into existing wireless communication systems, recent studies \cite{CSM,sun2024power} have explored practical IRS reflection designs by leveraging users' received signal (pilot/data) power measurements readily available in current communication systems, which are thus compatible with current cellular/WiFi protocols. It was shown that the conditional sample mean (CSM) method \cite{CSM}, with a sufficiently large number of reference signal received power (RSRP) measurement samples, can achieve the same signal-to-noise ratio (SNR) scaling order as under perfect CSI. However, CSM generally requires substantial power measurements due to the absence of channel recovery from power measurements. To address this, the authors in \cite{sun2024power} introduced channel recovery methods for power-measurement-based IRS reflection design. These studies, however, focus only on frequency-flat narrowband systems, leaving the more complex frequency-selective wideband channels unexplored. In wideband systems, IRS channel recovery is more challenging due to the increased channel coefficients caused by multi-path delay spread and their more intricate effects on RSRP measurements.

To tackle the above challenges, we propose a user power-/RSRP-measurement-based IRS channel estimation and reflection design framework for the wideband orthogonal frequency division multiplexing (OFDM) wireless system, as shown in Fig. \ref{fig0001}. First, we derived the user's average received signal power over all OFDM subcarriers in terms of the IRS channel autocorrelation matrix for the IRS-cascaded wideband channel, showing that it could be efficiently recovered based on the user's RSRP measurements over only a subset of subcarriers. To estimate this wideband channel autocorrelation matrix, we propose a single-layer neural network (NN)-enabled method utilizing the user's RSRP measurements. In particular, we reveal that under any given IRS passive reflection, the RSRP can be expressed as the output of a NN, which is the sum of the outputs of multiple single-layer subnetworks. Thus, the channel autocorrelation matrix can be recovered from the weights of the subnetworks trained via supervised learning. With the estimated channel autocorrelation matrix, the passive reflection coefficients of IRS are then optimized for maximizing the average received SNR over all OFDM subcarriers. Extensive simulation results verify that our proposed methods can outperform other benchmark schemes under wideband channels with significantly reduced power-measurement overhead and approach the average SNR upper bound assuming perfect CSI.

\emph{Notations}: $\boldsymbol{I}_M$ stands for the identity matrix of size $M \times M$, and $\textbf{0}_{M}$ denotes an all-zero vector of dimension $M$. $\mathbb{C}^{n \times m}$ and $\mathbb{R}^{n \times m}$ stand for the sets of the complex-valued and the real-valued matrices with size of $n \times m$, respectively. $\mathbb{N}$ stands for the set of nonnegative integers. $\left|\cdot\right|$ represents the cardinality of a set or the absolute value of a complex scalar. $\left\|\cdot\right\|$ stands for the Euclidean norm of a vector, and $\left\|\cdot\right\|_F$ represents the Frobenius norm of a matrix. $\Re\left(\cdot\right)$ and $\Im\left(\cdot\right)$ stand for the real part and the imaginary part of a complex-valued number/vector, respectively. $\jmath = \sqrt{-1}$ denotes the imaginary unit. ${\cal CN}{\left(\boldsymbol{0}_M,\sigma^2\boldsymbol{I}_M\right)}$ stands for the distribution of a circularly symmetric complex Gaussian (CSCG) random vector that has mean $\boldsymbol{0}_M$ and covariance $\sigma^2\boldsymbol{I}_M$. For a vector/matrix, $\left(\cdot\right)^T$ and $\left(\cdot\right)^H$ represent its transpose and its conjugate transpose, respectively. $\text{rank}(\cdot)$ computes the rank of a matrix, and $\text{Tr}(\cdot)$ returns the trace of a square matrix. $\boldsymbol{R}\succeq0$ indicates that $\boldsymbol{R}$ is positive semi-definite (PSD). $\text{diag}\left(\boldsymbol{x}\right)$ denotes a square diagonal matrix with vector $\boldsymbol{x}$ denoting the entries on its main diagonal. $\mathbb{E}\left[\cdot\right]$ represents the expectation of a random variable/vector.

\section{System Model}\label{Sec002}

As depicted in Fig. \ref{fig0001}, the IRS with $N$ reflecting elements is deployed to facilitate the wideband downlink transmission in an OFDM system. A single-antenna BS (or a multi-antenna BS with fixed transmit precoding) sends downlink messages to a single-antenna user. A central controller (e.g., the BS or a fusion center) exchanges control messages with the smart controller of the IRS to dynamically adjust the IRS reflection coefficients. Let $v_n = e^{\jmath \theta_n}, \ n \in {\cal N} \triangleq \left\{1,2, \cdots , N\right\}$, denote the $n$-th IRS reflecting element with $\theta_n$ denoting its phase shift. To simplify the hardware design and reduce costs, the phase shift of each IRS reflecting element is limited to a predefined finite set of values, uniformly distributed across the full phase range $\left(0,2\pi\right]$, i.e., ${\theta}_n \in \Phi_\mu \triangleq \{\omega, {2\omega}, \cdots ,2^\mu{\omega}\}, n \in {\cal N}$,
with $\mu$ denoting the number of IRS discrete phase-shift controlling bits, and $\omega \triangleq \frac{2\pi}{2^\mu}$.

In this paper, we assume quasi-static block-fading model for all links, which remain approximately unchanged within each channel coherence block\cite{Zhengbx2020,wu2023intelligent,mei}. For the considered OFDM system, let $K_1$, $K_2$ and $K_3$ denote the maximum number of multi-path delay taps of the BS-user, BS-IRS and IRS-user channels, respectively.
The baseband equivalent BS-user channel, that from the BS to the $n$-th IRS reflecting element and that from the $n$-th IRS reflecting element to the user are denoted by $\boldsymbol{\bar{f}} = \left[ f_{1},f_{2}, \cdots , f_{K_1}\right]^T \in \mathbb{C}^{K_1 \times 1}$, $\boldsymbol{q}_n = \left[ q_{n,1},q_{n,2}, \cdots , q_{n,K_2}\right]^T \in \mathbb{C}^{K_2 \times 1}$ and $\boldsymbol{b}_n = \left[ b_{n,1},b_{n,2}, \cdots , b_{n,K_3}\right]^T \in \mathbb{C}^{K_3 \times 1}$, respectively. Accordingly, the time-domain cascaded BS-IRS-user channel (without involving the effect of IRS's phase shifts ${v}_n$) via the $n$-th IRS reflecting element can be computed as $\boldsymbol{\bar g}_n = \boldsymbol{q}_n * \boldsymbol{b}_n \in \mathbb{C}^{K_r \times 1}$, where $K_r = K_2+K_3-1$ denotes the number of delay taps for the IRS cascaded channel. Define $K \triangleq \max\left({K_1,K_r}\right)$ as the maximum number of delay taps for both the BS-user direct channel and the cascaded BS-IRS-user channel. Define $M$ as the number of OFDM subcarriers with $M \gg K$. Let $\boldsymbol{f} = \left[ \boldsymbol{\bar{f}}^T, \textbf{0}_{M-K_1}^T\right]^T \in \mathbb{C}^{M \times 1}$ and $\boldsymbol{g}_n = \left[ {\boldsymbol{\bar g}}_n^T, \textbf{0}_{M-K_r}^T\right]^T \in \mathbb{C}^{M \times 1}$ denote the zero-padded BS-user direct channel and cascaded BS-IRS-user channel via the $n$-th IRS reflecting element, respectively. As such, the superimposed channel impulse response (CIR) of the wideband channel is given by
\begin{equation}
  \tilde{\boldsymbol{h}} = \sum_{n=1}^{N}{v_n\boldsymbol{g}_n} + \boldsymbol{f}.
\end{equation}
Let $\boldsymbol{G} = \left[\boldsymbol{f}, \boldsymbol{g}_1,\boldsymbol{g}_2,\cdots,\boldsymbol{g}_N\right]\in \mathbb{C}^{M \times \left(N+1\right)}$ denote the CIR matrix of the wideband channel, and $\boldsymbol{v} \triangleq \left[1, v_1, v_2, \cdots , v_N\right]^T$ denote the extended IRS passive reflection vector. As such, the above CIR can be rewritten by $\tilde{\boldsymbol{h}} = {\boldsymbol{G}{\boldsymbol{v}}}$.
Based on the CIR of the wideband channel, the channel frequency response (CFR) over all of the $M$ subcarriers is thus given by
\begin{equation}\label{Eqs00206}
  \boldsymbol{h} = \boldsymbol{F}_M\tilde{\boldsymbol{h}} = \boldsymbol{F}_M{\boldsymbol{G}{\boldsymbol{v}}},
\end{equation}
where $\boldsymbol{F}_M$ denotes the discrete Fourier transform (DFT) matrix of size $M \times M$. For convenience, an equal power allocation is applied over the $M$ subcarriers, and the baseband received signal in the frequency-domain is given by $\boldsymbol{y} = \boldsymbol{X}\boldsymbol{h} + \boldsymbol{z}$, where $\boldsymbol{y} \triangleq \left[y_1,y_2, \cdots ,y_M\right]^T \in \mathbb{C}^{M \times 1}$ with $y_m$ denoting the received signal at the $m$-th subcarrier, $\boldsymbol{X}=\text{diag}\left(\boldsymbol{x}\right) \in \mathbb{C}^{M \times M}$ denotes the diagonal matrix of an OFDM symbol $\boldsymbol{x} \in \mathbb{C}^{M \times 1}$ with $\mathbb{E}\left[\left\|\boldsymbol{x}\right\|^2\right]=P$ denoting the BS's transmit power, and $\boldsymbol{z} \triangleq \left[z_1,z_2,\cdots,z_M\right]^T \in \mathbb{C}^{M \times 1} \sim {\cal CN}{\left(\textbf{0}_M,\sigma^2\boldsymbol{I}_M\right)}$ represents the receiver noise vector with $\sigma^2$ denoting the average noise power. Sicne $\boldsymbol{z}$ is independent of both $\boldsymbol{X}$ and $\boldsymbol{h}$, the average received signal power over $M$ subcarriers is given by
\begin{equation}\label{Eqs00208}
\begin{split}
  {p}\left(\boldsymbol{v}\right) = \frac{1}{M}\mathbb{E}\left[\left\|\boldsymbol{X}{\boldsymbol{h}}\right\|^2\right] + \sigma^2.
\end{split}
\end{equation}
By substituting (\ref{Eqs00206}) into (\ref{Eqs00208}), we have
\begin{align}\label{Eqs002010}
{p}\left(\boldsymbol{v}\right) & = \frac{1}{M}\mathbb{E}\left[{\boldsymbol{v}^H}{\boldsymbol{G}^H}\boldsymbol{F}_M^H\boldsymbol{X}^H\boldsymbol{X}\boldsymbol{F}_M{\boldsymbol{G}{\boldsymbol{v}}}\right] + \sigma^2
 \nonumber \\
  & = {\boldsymbol{v}}^H{\boldsymbol{R}}{\boldsymbol{v}} + \sigma^2,
\end{align}
where we have utilized the fact that $\mathbb{E}\left[\boldsymbol{F}_M^H\boldsymbol{X}^H\boldsymbol{X}\boldsymbol{F}_M\right] = P \boldsymbol{I}_M$, and ${\boldsymbol{R}}\triangleq \frac{P}{M}{\boldsymbol{G}}^H\boldsymbol{G} \in \mathbb{C}^{(N+1) \times (N+1)}$ denotes the autocorrelation matrix of $\boldsymbol{G}$ (scaled by the power factor $P/{M}$). Note that the channel autocorrelation matrix ${\boldsymbol{R}}$ is Hermitian and PSD (i.e., ${\boldsymbol{R}} \succeq 0$).

In this paper, we aim to optimize the IRS reflection to maximize the average received SNR over all subcarriers at the user in the wideband system subject to the discrete IRS phase shift constraints. Accordingly, the IRS reflection design problem is formulated as
\begin{subequations}\label{Eqs007}
\begin{align}
  (\text{P1}): \ & \max_{\boldsymbol{v}}\frac{{{\boldsymbol{v}}^H{\boldsymbol{R}}{\boldsymbol{v}}}}{\sigma^2} \label{00701} \\
    & \text{s.t.} \ {\theta}_n \in \Phi_\mu, n \in {\cal N}. \label{00702}
\end{align}
\end{subequations}
Acquiring an accurate channel autocorrelation matrix $\boldsymbol{R}$ for solving (P1) is practically challenging due to the limited signal processing capabilities of the IRS and the high-dimensional channel parameters in IRS-aided wideband systems. To tackle this problem, we propose a NN-based IRS channel autocorrelation estimation method by leveraging user power measurements for wideband IRS reflection design, as detailed next.

\section{NN-enabled Channel Estimation Based on RSRP Measurements }\label{Sec003}

\subsection{RSRP Measurement }

\begin{figure}[t]
  \centering
  {\includegraphics[width=0.3695\textwidth]{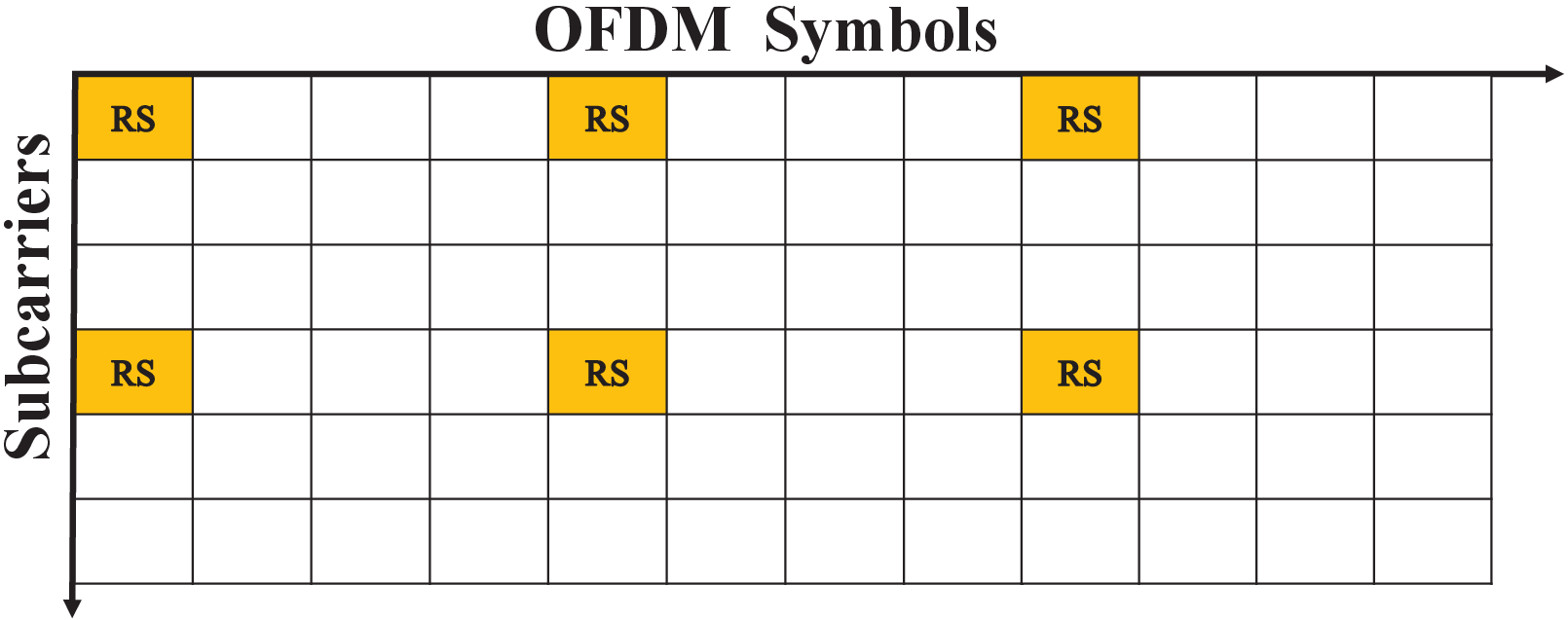}}
  \caption{An example of inserted RSs over OFDM symbols with $Q=3$, $M_0=2$, and $M=6$. }
\label{fig000301}
\vspace{-9pt}
\end{figure}

In morden wireless communication systems (cellular/WiFi), user terminals are able to measure their received signal power, e.g., RSRP. As shown in Fig. \ref{fig000301}, the reference signals (RSs) are inserted in a subset of the $M$ subcarriers for each OFDM symbol for high efficiency of time-frequency resource utilization\cite{3GPPTS361,3GPPTS36}. By leveraging the RSRP measurement capability in current cellular/WiFi systems for wideband IRS channel estimation, the central controller configures the IRS with a set of randomly (subject to (\ref{00702})) generated phase shifts for reflection training and RSRP measurements. For each IRS reflection set, the user terminal measures the received power of RSs in $Q$ (non-consecutive) OFDM symbols (see Fig. \ref{fig000301}) in one channel coherence block, where ${\cal Q} \triangleq \left\{1,2,\cdots,Q\right\}$ denotes the set of these OFDM symbols. \footnote{In practice, even if OFDM channel varies over channel coherence blocks, the channel autocorrelation matrix $\boldsymbol{R}$ stays approximately constant over a longer time as it is mainly determined by the line-of-sight (LoS) channel path or non-LoS (NLoS) channel paths due to static scatterers in the environment \cite{LITOFDM,goldsmith2005wireless}. Thus, the RSRP-based estimation of $\boldsymbol{R}$ is practically robust to small-scale channel variations in space/time.}
For each OFDM symbol, the RSs are assumed to be \emph{uniformly} inserted into $M_0, (M_0 < M)$ subcarriers for power measurement with ${\cal M}_0 \triangleq \left\{m_1,m_2,\cdots,m_{M_0}\right\}$ denoting the index set of the subcarriers inserted with RSs.
Accordingly, the RSRP under any given IRS reflection vector $\boldsymbol{v}$ is given by
\begin{equation}\label{Eqs00301}
\begin{split}
  \bar{p}\left(\boldsymbol{v}\right) = \frac{1}{QM_0}\sum_{q \in {\cal{Q}}}{\sum_{m \in {\cal{M}}_0}{\left|{x}_m{{h}_m} + {z}_m(q)\right|^2}},
\end{split}
\end{equation}
where ${z}_m(q)$ denotes the noise at the $m$-th subcarrier in the $q$-th OFDM symbol at the user receiver. Since the symbol rate of OFDM is usually much higher than the reflection switching rate of the IRS, we have $Q \gg 1$ in practice. As such, the RSRP measurement in (\ref{Eqs00301}) should approach the average received power over RSs in ${\cal M}_0$ with a sufficiently large $Q$, i.e.,
\begin{equation}\label{Eqs0030101}
  \bar{p}\left(\boldsymbol{v}\right) \approx \tilde{p}\left(\boldsymbol{v}\right) \triangleq \frac{1}{M_0}\mathbb{E}\left[\left\|\boldsymbol{\bar{X}}{\boldsymbol{\bar{h}}} + \boldsymbol{\bar{z}}\right\|^2\right],
\end{equation}
where $\boldsymbol{\bar{X}}=\text{diag}\left(x_{m_1},x_{m_2},\cdots,x_{m_{M_0}}\right) \in \mathbb{C}^{M_0 \times M_0}$ denotes the diagonal OFDM symbol matrix over the $M_0$ subcarriers in ${\cal M}_0$, $\boldsymbol{\bar{z}} \triangleq \left[z_{m_1},z_{m_2},\cdots,z_{m_{M_0}}\right]^T \in \mathbb{C}^{M_0 \times 1} \sim {\cal CN}{\left(\textbf{0}_{M_0},\sigma^2\boldsymbol{I}_{M_0}\right)}$ denotes the receiver noise vector over the $M_0$ subcarriers, and ${\boldsymbol{\bar{h}}}=\boldsymbol{\bar{F}}\boldsymbol{G}\boldsymbol{v}$ denotes the CFR of the OFDM channel over the $M_0$ subcarriers with $\boldsymbol{\bar{F}} \in \mathbb{C}^{M_0 \times M}$ denoting a partial DFT matrix. Specifically, the $i$-th row of $\boldsymbol{\bar{F}}$ is taken from the $m_i$-th row of $\boldsymbol{F}$, $m_i \in {\cal M}_0$.
With respect to (\ref{Eqs0030101}), we present the following lemma to reveal the relationship between the RSRP measurements and the average received signal power over all of the $M$ subcarriers, i.e., (\ref{Eqs002010}).
Note that in LTE and NR systems, the number of RS-inserted subcarriers is larger than the number of delay taps\cite{booksofdmsntenna}, i.e., $M_0 \geq K$.

\begin{lemma}\label{lemma0}
When $M_0 \geq K$, the RSRP measurement in (\ref{Eqs0030101}) based on a subset of subcarriers in ${\cal M}_0$ is equal to the average received signal power over all the $M$ subcarriers in (\ref{Eqs002010}), i.e.,
\begin{equation}\label{Eqs00302}
\begin{split}
  {\tilde{p}}\left(\boldsymbol{v}\right) = {p}\left(\boldsymbol{v}\right) = {\boldsymbol{v}}^H{\boldsymbol{R}}{\boldsymbol{v}} + \sigma^2.
\end{split}
\end{equation}
\end{lemma}
\begin{proof}
    See Appendix.
\end{proof}

Lemma \ref{lemma0} implies that although the RSRP measurement is performed at only a subset of subcarriers, it can characterize the average received signal power over all of the $M$ subcarriers, suggesting that the wideband channel autocorrelation matrix can be extracted from the RSRP measurements. To this end, the central controller collects the RSRP measurements ${{\cal{P}} \triangleq \{\bar{p}(\boldsymbol{v}_1),\bar{p}(\boldsymbol{v}_2),\cdots,\bar{p}(\boldsymbol{v}_L)\}}$ from the user side under $L$ different IRS reflection vectors $\boldsymbol{v}_1,\boldsymbol{v}_2,\cdots,\boldsymbol{v}_L$, for IRS channel autocorrelation matrix estimation, as presented next.

\subsection{NN-Empowered Channel Autocorrelation Matrix Estimation}

The channel autocorrelation matrix $\boldsymbol{R} \in \mathbb{C}^{(N+1) \times (N+1)}$ is a high-dimensional Hermitian and PSD matrix, whose estimation entails the complexity in the order of ${\cal O}\left(N^2\right)$. Fortunately, the rank of $\boldsymbol{R}$, which is equal to the maximum number of delay taps, $K$, is in general much smaller than the size of channel autocorrelation matrix $\boldsymbol{R}$, i.e., $K \ll (N+1)$, for an IRS with numerous reflecting elements, $N$. By exploiting this sparsity property of $\boldsymbol{R}$, we can express $\boldsymbol{R}$ equivalently using a small number ($K$) of basis vectors via matrix spectral decomposition, thereby transforming the high-dimensional channel autocorrelation matrix estimation problem into a lower-dimensional basis-vector estimation problem, thus significantly reducing the computational complexity. To achieve this purpose, we present the following lemma.
\begin{lemma}
\label{lemma1}
For any $N$-dimensional Hermitian matrix $\boldsymbol{R}$ with $\text{rank}\left({\boldsymbol{R}}\right) = K$, there always exist $K$ basis vectors, $\boldsymbol{a}_k \in \mathbb{C}^{(N+1) \times 1}$, $1 \le k \le K$, such that
\begin{equation}\label{Eqs00303}
   \boldsymbol{R} = \sum_{k=1}^{K}{\boldsymbol{a}_k}{\boldsymbol{a}^H_k}.
\end{equation}
\end{lemma}
The existence of (\ref{Eqs00303}) can be readily proofed by constructing a set of basis vectors through spectral decomposition to the Hermitian and PSD matrix $\boldsymbol{R}$. Lemma \ref{lemma1} illustrates that $\boldsymbol{R}$ can always be recovered as the sum of $K$ rank-one matrices shown in the right-hand side of (\ref{Eqs00303}). Accordingly, the noiseless received signal power can be represented as
\begin{equation}\label{Eqs00308}
  \boldsymbol{v}^H{\boldsymbol{R}}{\boldsymbol{v}}  =\sum_{k=1}^{K}{\left|{\boldsymbol{v}}^H\boldsymbol{a}_k\right|^2}.
\end{equation}
Notably, the noiseless received signal power in (\ref{Eqs00308}) can be predicted by a single-layer NN. Particularly, the NN takes the IRS passive reflection vector $\boldsymbol{v}$ and (\ref{Eqs00308}) as its input and output, respectively. Since (\ref{Eqs00308}) is the sum of $K$ squared amplitude values, i.e., $\lvert \boldsymbol{v}^H \boldsymbol{a}_k \rvert^2, k=1,2,\cdots, K$, the NN can be divided into $K$ subnetworks with the weights of the $k$-th subnetwork corresponding to the basis vector $\boldsymbol{a}_k$, as shown in Fig. \ref{fig001}. Moreover, let $\boldsymbol{u}=\left[{\Re}\left(\boldsymbol{v}^T\right), {\Im}\left(\boldsymbol{v}^T\right)\right]^T \in \mathbb{R}^{\left(2N+2\right) \times 1}$ denote the real-valued IRS passive reflection vector, and $\boldsymbol{B}_k$ denote the real-valued basis matrix, i.e.,
\begin{equation}\label{Eqs003010}
  \boldsymbol{B}_k = \left[ {\begin{array}{*{20}{c}}
{ {{\mathop{\rm \Re}\nolimits} \left( \boldsymbol{{a}}_k \right)} }&{{{{\mathop{\rm \Im}\nolimits} \left( \boldsymbol{{a}}_k \right)}}}\\
{ {{{\mathop{\rm \Im}\nolimits} \left( \boldsymbol{{a}}_k \right)}}}&{{-{{\mathop{\rm \Re}\nolimits} \left( \boldsymbol{{a}}_k \right)}}}
\end{array}} \right]\in \mathbb{R}^{\left({2N+2}\right) \times 2}.
\end{equation}
As such, we can express (\ref{Eqs00308}) in the real-value domain as
\begin{equation}\label{Eqs00309}
  \sum_{k=1}^{K}{\left|{\boldsymbol{v}}^H\boldsymbol{a}_k\right|^2} = \sum_{k=1}^{K}{\left\|\boldsymbol{u}^T{\boldsymbol{B}_k}\right\|^2}.
\end{equation}
Based on (\ref{Eqs00309}), the NN takes the real-valued IRS passive reflection vector $\boldsymbol{u}$ as its input and aims to recover the noiseless received signal power at its output. Particularly, the $K$ subnetworks share the same input $\boldsymbol{u}$ and aim to recover ${\left\|\boldsymbol{u}^T{\boldsymbol{B}_k}\right\|^2}, k=1,2, \cdots ,K$, at their individual outputs, respectively. Define $\boldsymbol{W}_{k} \in \mathbb{R}^{\left({2N+2}\right) \times 2}$ as the weights of the $k$-th subnetwork. The values of two neurons in the hidden layer of the $k$-th subnetwork can be computed by $\boldsymbol{e}^T_k =\boldsymbol{u}^T{\boldsymbol{W}_k}$,
with $\boldsymbol{e}_k \triangleq \left[{e}_{k,1} , {e}_{k,2}\right]^T \in \mathbb{R}^{2 \times 1}$. A squared norm function is adopted as the activation function in the $k$-th subnetwork to compute ${\left\|\boldsymbol{e}_k\right\|^2}$. The output layer of the entire NN is given by the sum of the outputs of the $K$ subnetworks, i.e.,
\begin{equation}\label{Eqs003012}
  \hat p \left(\boldsymbol{v}\right) = \sum_{k=1}^{K}{\left\|\boldsymbol{e}_k\right\|^2} = \sum_{k=1}^{K}{\left\|\boldsymbol{u}^T{\boldsymbol{W}_k}\right\|^2}.
\end{equation}

\begin{figure}[t]
  \centering
  {\includegraphics[width=0.4361\textwidth]{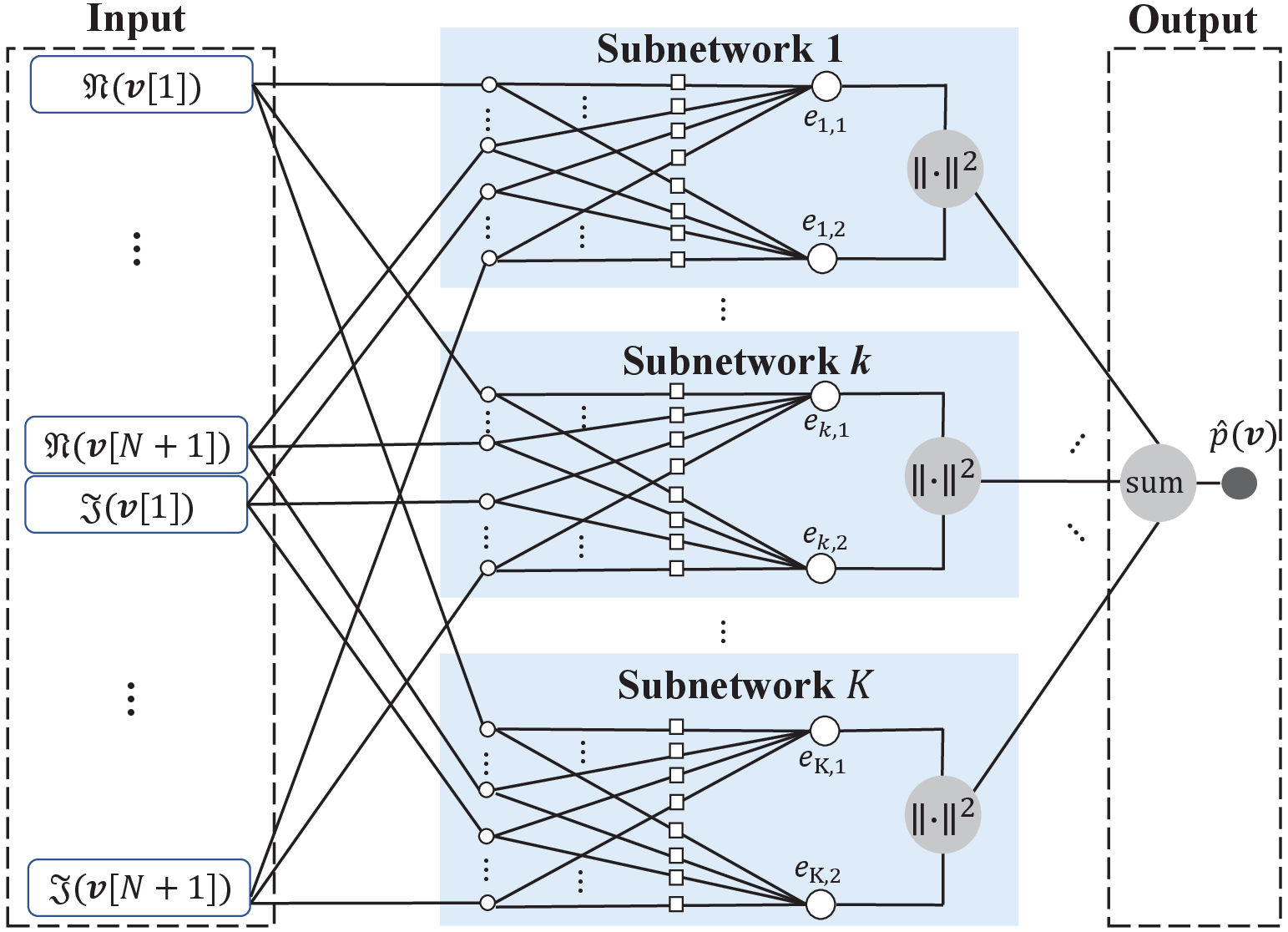}}
  \caption{NN structure for wideband channel autocorrelation estimation. }
\label{fig001}
\vspace{-16pt}
\end{figure}

Comparing (\ref{Eqs00309}) with (\ref{Eqs003012}), it follows that if $\boldsymbol{W}_{k}=\boldsymbol{B}_k, \forall k$, we have $\hat p\left(\boldsymbol{v}\right) = \boldsymbol{v}^H{\boldsymbol{R}}{\boldsymbol{v}}$. Inspired by this, the real-valued basis matrix $\boldsymbol{B}_k$ can be recovered by training the NN's weight matrix $\boldsymbol{W}_{k}$. To this end, the weight matrix $\boldsymbol{W}_{k}, \forall k,$ is designed to share a  structure similar to that of $\boldsymbol{B}_k$, i.e.,
\begin{equation}\label{Eqs001013}
  \boldsymbol{W}_{k} = \left[ {\begin{array}{*{20}{c}}
{ {\boldsymbol{w}_{k,1}} }&{{\boldsymbol{w}_{k,2}}}\\
{ {{\boldsymbol{w}_{k,2}}}}&{{-{\boldsymbol{w}_{k,1}}}}
\end{array}} \right] \in \mathbb{R}^{\left({2N+2}\right) \times 2} ,
\end{equation}
where $\boldsymbol{w}_{k,1} \in \mathbb{R}^{(N+1) \times 1}$ and $\boldsymbol{w}_{k,2} \in {\mathbb{R}}^{(N+1) \times 1}$ correspond to the real and imaginary parts of $\boldsymbol{a}_k$, i.e., ${{{\mathop{\rm \Re}\nolimits} \left( \boldsymbol{{a}}_k \right)}}$ and ${{{{\mathop{\rm \Im}\nolimits} \left( \boldsymbol{{a}}_k \right)}}}$, respectively.
Let $\boldsymbol{w}_{k} = \boldsymbol{w}_{k,1} + \jmath \boldsymbol{w}_{k,2}$ denote the complex weight vector of the $k$-th subnetwork. By substituting (\ref{Eqs001013}) into (\ref{Eqs003012}) and leveraging $\boldsymbol{u}=\left[{\Re}\left(\boldsymbol{v}^T\right), {\Im}\left(\boldsymbol{v}^T\right)\right]^T$, we obtain
\begin{equation}\label{Eqs0001016}
  \hat p(\boldsymbol{v}) = \sum_{k=1}^{K}{\left|\boldsymbol{v}^H\boldsymbol{w}_k\right|^2}=\boldsymbol{v}^H\left(\sum_{k=1}^{K}{\boldsymbol{w}_k\boldsymbol{w}_k^H}\right)\boldsymbol{v}.
\end{equation}
If ${\hat p}\left(\boldsymbol{v}\right) = {\boldsymbol{v}}^H{\boldsymbol{R}}{\boldsymbol{v}}$ holds for any $\boldsymbol{v}$, then $\boldsymbol{v}^H\left(\sum_{k=1}^{K}{\boldsymbol{w}_k\boldsymbol{w}_k^H}- \boldsymbol{R}\right)\boldsymbol{v} = 0$ should hold for any $\boldsymbol{v}$, and thus the IRS channel autocorrelation matrix can be estimated by
\begin{equation}\label{Eqs001015}
  \boldsymbol{\hat{R}} = \sum_{k=1}^{K}{\boldsymbol{w}_{k}\boldsymbol{w}_{k}^H}.
\end{equation}
It follows from the above that we can reconstruct $\boldsymbol{R}$ by training the NN via supervised learning. Specifically, the IRS training reflection sets $\boldsymbol{v}_1,\cdots,\boldsymbol{v}_L$ and their corresponding RSRP measurements $\cal{P}$ are used as the input of the NN and the training labels, respectively. To enhance the generalization ability of the NN, we randomly select $L_1 (L_1<L)$ entries of $\cal P$ as the training set, with the remaining data utilized as the validation set. The loss function is defined as the mean square error (MSE) between the RSRP measurement and the NN's output, i.e.,
\begin{equation}
    {\cal L}_{\boldsymbol{W}} = \frac{1}{L_1}\sum_{l=1}^{L_1}{\left(\bar p(\boldsymbol{v}_l) - \sigma^2 - \hat p(\boldsymbol{v}_l)\right)^2}.
\label{losssn}
\end{equation}
The NN in Fig. \ref{fig001} can be trained by minimizing (\ref{losssn}) using the stochastic gradient descent algorithm\cite{SGD}.

\subsection{IRS Reflection Design}
Based on the estimate in (\ref{Eqs001015}), we can optimize the IRS passive reflection by solving problem (P1) via replacing the channel autocorrelation matrix therein with its estimate, ${\boldsymbol{\hat R}}$. However, optimization problem (P1) is non-convex because of the discrete phase-shift constraints (\ref{00702}). To solve (P1) efficiently, we combine the semidefinite relaxation (SDR) technique with the successive refinement (SF) strategy\cite{WdIRS}. Define ${\boldsymbol{V}=\boldsymbol{v}\boldsymbol{v}^H}$ as the autocorrelation matrix of reflection vector ${\boldsymbol{v}}$, which is a PSD matrix, i.e., $\boldsymbol{V} \succeq 0$. As such, problem (P1) is reformulated as
\begin{subequations}\label{Eqs00505}
\begin{align}
  \text{(P2):} &\max_{{\boldsymbol{V}}} \ \text{Tr}\left(\boldsymbol{\hat R}\boldsymbol{V}\right) \label{O03016} \\
 & \ \ \text{s.t.} \ {\theta}_n \in \Phi_\mu, \ n=1,2,\cdots,N,\label{C0301604} \\
 & \ \ \ \ \ \ {\boldsymbol{V}}\succeq 0,\label{C0301603} \\
 & \ \ \ \ \ \ \ \text{rank}({\boldsymbol{V}})=1\label{O030165} .
\end{align}
\end{subequations}
Next, we relax problem (P2) into a convex semidefinite programming (SDP) problem by replacing the constraint (\ref{C0301604}) with $\boldsymbol{V}_{n,n}=1, n=1, 2, \cdots , N+1$ and removing the constraint (\ref{O030165}). As such, the optimal solution to this SDP problem can be obtained by utilizing the interior-point algorithm\cite{cvx}. However, the achieved solution to the above SDP problem may not satisfy (\ref{C0301604}) or (\ref{O030165}). To address this problem, we perform Gaussian randomization jointly with phase quantization to achieve a feasible solution to (P2). Finally, we use the SF method to refine the above solution.

\section{Simulation Results}\label{sec004}

The simulation setup is shown in Fig. \ref{fig0001}, where an IRS is deployed parallel to the $y$-$z$ plane in a three-dimensional Cartesian coordinate system (in meter). The IRS is equipped with a uniform planar array (UPA) consisting of $N = 8 \times 4 = 32$ reflecting elements with inter-element spacing equal to half wavelength. The reference point (the bottom-left reflecting element) of the IRS is located at $(-2,-1,0)$. The BS and the receiver are deployed at $(35,-20,15)$ and $(0,1,0)$, respectively. Let $\beta_1, \beta_2$ and $\beta_3$ denote the distance-dependent path loss (in dB) of the BS-user, BS-IRS and IRS-user channels, which are given by $\beta_1 =  33 + 37\text{log}_{10}(d_1), \beta_2 =  30 + 20\text{log}_{10}(d_2), \beta_3 =  30 + 20\text{log}_{10}(d_3)$\cite{YWJSAC,CSM,3GPPTS25996}, with $d_1,d_2$ and $d_3$ denoting the distance between the BS and the user, that between the BS and the IRS, and that between the IRS and the user, respectively.

The OFDM system is configured with $M = 128$ subcarriers, with the number of delay taps for the BS-user, BS-IRS and IRS-user channels set to $K_1=4$, $K_2=4$ and $K_3=3$, respectively. Considering that the BS is far away from the IRS/user, the BS-IRS and BS-user frequency-selective channels are simulated using an exponentially decaying power delay profile given by $\zeta_{i , k }=\frac{1}{\eta_{i , k}}e^{-\varepsilon(k-1)}$, $k=1,2, \cdots , K_i, i=1,2$, where $\eta_{i,k}=\sum_{k=1}^{K_i}{e^{-\varepsilon(k-1)}}$ denotes the normalization factor with $\varepsilon$ denoting the exponential decay factor\cite{Powerechannel0,riihonen2010generalized}. The CIR coefficients for each link are generated based on independent and identically distributed (i.i.d.) Rayleigh fading with average power $10^{-\beta_i/10}\zeta_{i,k}, k=1,2, \cdots , K_i, i=1,2$\cite{Zhengbx2020}. Besides, we assume frequency-selective Rician fading for the IRS-user link. Specifically, the first tap of the IRS-user link is modeled as the line-of-sight (LoS) component, while the other taps are modeled as i.i.d. non-LoS (NLoS) Rayleigh fading components\cite{Zhengbx2020}. Let $\kappa$ denote the Rician factor, i.e., the power ratio of the dominant LoS component to that of all the NLoS fading components. In the simulation, we set $\kappa=7$ dB\cite{3GPPTS25996}. The number of RS-inserted OFDM symbols is set to $Q=30$, and the number of RSs inserted in each OFDM symbol is set to $M_0=64$\cite{3GPPTS36}. The power decaying factor in the OFDM channel is set to $\varepsilon=2$. The noise power is $\sigma^2=-90$ \text{dBm} and the BS's transmit power is $P=30$ \text{dBm}. The number of bits for controlling IRS reflection phase shifts is set to $\mu = 2$, if not specified otherwise.

\begin{figure}[t]
\hspace{0.2336cm}
\centering
{\includegraphics[width=0.4159\textwidth]{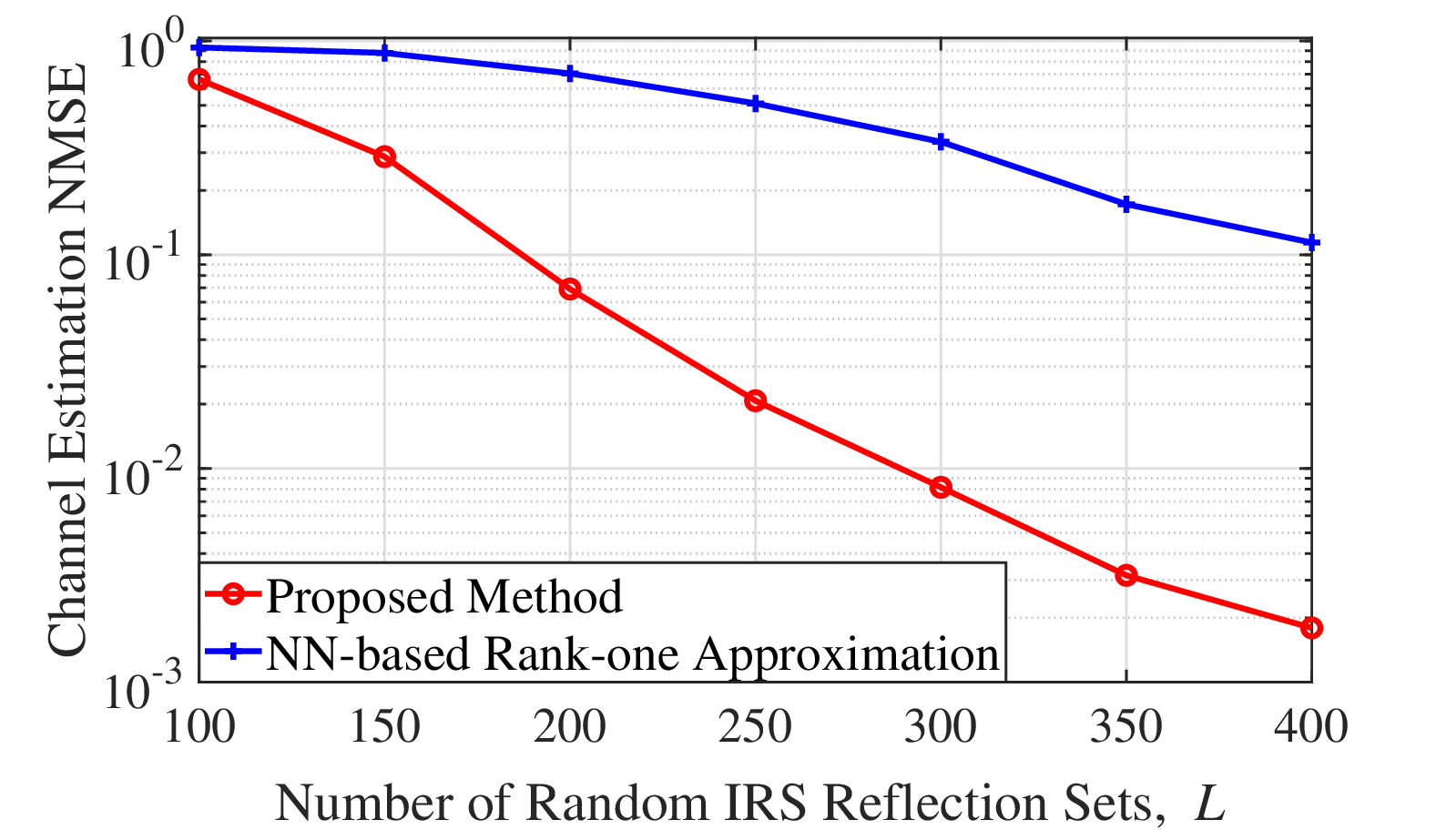}}
  \caption{IRS channel autocorrelation estimation NMSE versus the number of random IRS reflection sets, $L$. }
\label{fig0401}
\end{figure}

We evaluate the normalized MSE (NMSE) performance of the proposed NN-enabled IRS channel autocorrelation estimation algorithm. In particular, the NMSE between the estimated channel autocorrelation matrix and the actual one is given by
$\text{NMSE} = \frac{\left\|\boldsymbol{\hat R} - \boldsymbol{R}\right\|_F^2}{\left\|\boldsymbol{R}\right\|_F^2}$. Note that the proposed NN-enabled channel estimation method reduces to the NN-based narrow-band IRS channel estimation method\cite{sun2024power} by setting the number of subnetworks to 1. Thus, the NMSE by the proposed NN-enabled channel estimation with $1$ subnetwork (labeled as ``NN-based Rank-one Approximation'') are include as a benchmark.

Fig. \ref{fig0401} shows the NMSE of IRS channel estimation versus the number of random IRS reflection sets, $L$. It depicts that the NMSE decreases by enlarging $L$, as expected. In addition, the proposed NN-enabled channel estimation algorithm achieves much lower NMSE than the benchmark of Rank-one Approximation, owing to its ability to generate a higher-rank approximation of the channel autocorrelation matrix under the wideband setup.

Next, we compare the average received SNRs at the user in (\ref{00701}) by different IRS passive reflection designs. In particular, we adopt the random-max sampling (RMS)\cite{CSM} and CSM \cite{CSM} methods as benchmark schemes, as both of them utilize user power measurements for optimizing the IRS passive reflection. Specifically, the RMS method selects the IRS reflection vector that yields the maximum received signal power among all RSRP measurements as the optimized one. While the CSM method in \cite{CSM} computes the sample mean of RSRP measurements conditioned on $\theta_n=\psi, \psi \in \Phi_\mu$, and selects the phase shift that maximizes the CSM for each reflecting element. Furthermore, the IRS reflection design using perfect CSI (i.e., by substituting the perfect knowledge of ${\boldsymbol{R}}$ into (\ref{O03016}) for solving problem (P2)) is also included as the performance upper bound on the achievable average received SNR (labeled as ``Upper Bound'').
\begin{figure}
  \centering
  \subfigure[$\mu=1$]{\includegraphics[width=0.24\textwidth]{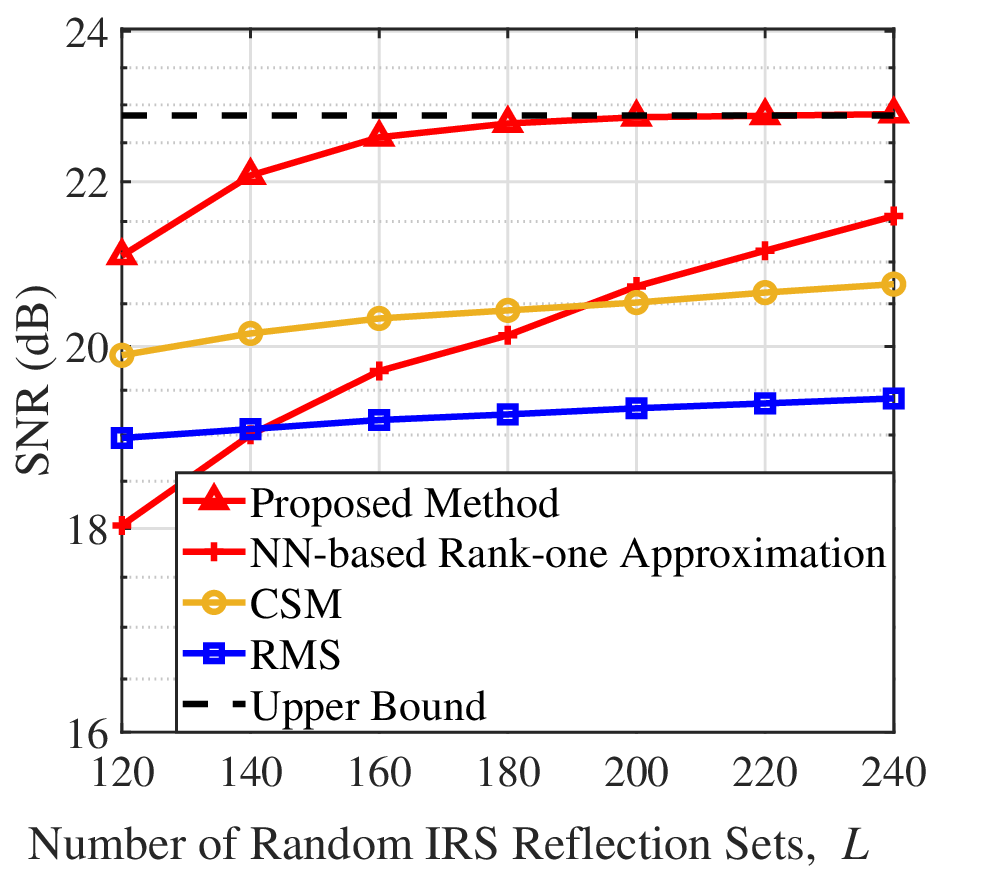}}
  \subfigure[$\mu=2$]{\includegraphics[width=0.24\textwidth]{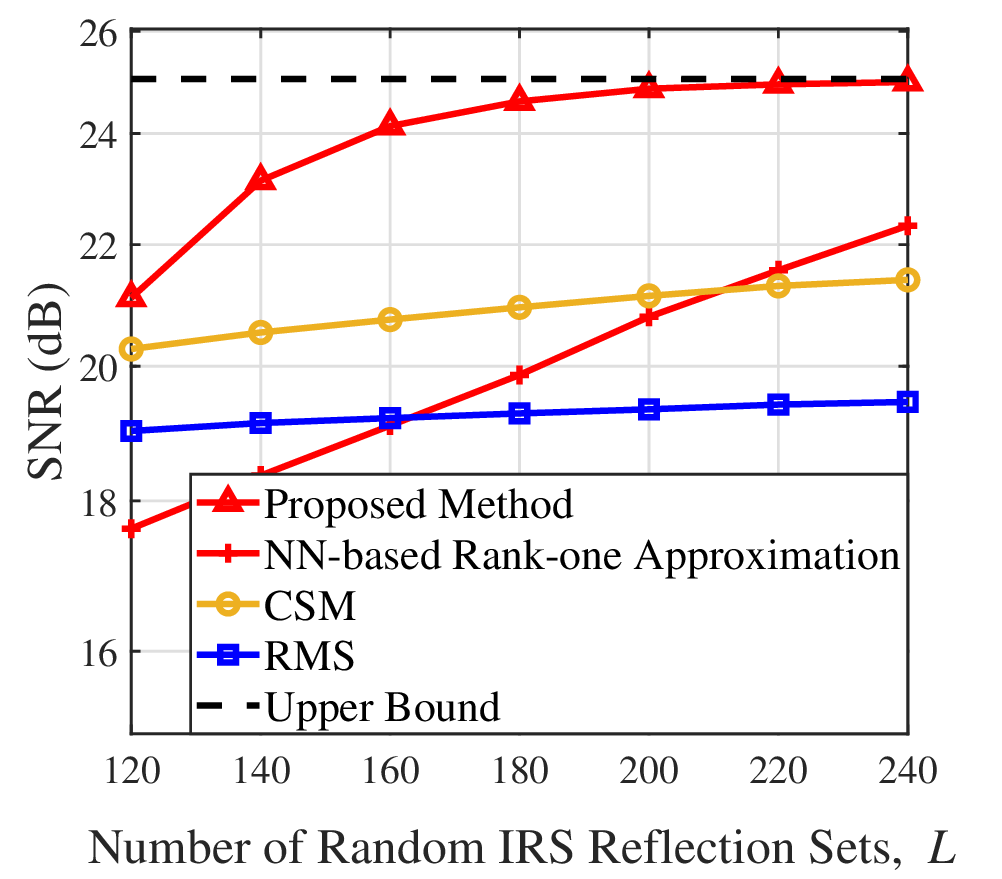}}
  \caption{Average SNR performance versus the number of random IRS reflection sets, $L$.}
\label{fig0005}
\end{figure}

Figs. \ref{fig0005}(a) and \ref{fig0005}(b) show the average received SNRs achieved by different IRS passive reflection designs versus the number of random IRS passive reflection sets, $L$, under $\mu = 1,2$, respectively. The results depict that the received SNR achieved by the proposed method is larger than that by other benchmark schemes. By further increasing $L$, the achieved SNR performance of the proposed algorithm can approach the upper bound assuming perfect CSI. Furthermore, a significant SNR improvement by the proposed scheme is observed after increasing $\mu$, owing to the improved resolution of IRS phase shift for both IRS channel autocorrelation matrix recovery and IRS reflection optimization.

\section{Conclusion}\label{Sec006}

In this paper, a NN-empowered wideband IRS channel autocorrelation estimation and IRS reflection design framework based on user power measurements was proposed for OFDM communication systems by utilizing a single-layer NN structure with its weights representing the autocorrelation matrix of the wideband IRS channel. Simulation results demonstrated that the proposed wideband IRS channel autocorrelation matrix recovery and IRS reflection optimization framework substantially outperforms other power-measurement-based benchmarks and approaches the performance upper bound. Our proposed wideband IRS channel autocorrelation matrix estimation design can be generalized to multi-IRS and/or multiuser communication systems, which are worth studying in future work.

\appendix
By taking advantage of the fact that $\mathbb{E}\left[\boldsymbol{\bar{X}}{\boldsymbol{\bar{h}}}\boldsymbol{\bar{z}}\right]=0$ and ${\boldsymbol{\bar{h}}}=\boldsymbol{\bar{F}}\boldsymbol{G}\boldsymbol{v}$, (\ref{Eqs0030101}) can be simplified as
\begin{equation}\label{Eqs0A002}
\begin{split}
\tilde{p}\left(\boldsymbol{v}\right) & = \frac{\mathbb{E}\left[\left\|\boldsymbol{\bar{X}}{\boldsymbol{\bar{h}}}\right\|^2\right]}{M_0} + \sigma^2 \\
&  = \frac{P{\boldsymbol{v}^H\boldsymbol{G}^H\boldsymbol{\bar{F}}^H\boldsymbol{\bar{F}}\boldsymbol{G}\boldsymbol{v}}}{M_0M} + \sigma^2,
\end{split}
\end{equation}
where we have utilized the fact that $\mathbb{E}\left[\boldsymbol{\bar{X}}^H\boldsymbol{\bar{X}}\right] = \frac{P}{M}\boldsymbol{I}_{M_0}$. It is worth noting that the partial DFT matrix $\boldsymbol{\bar{F}}$ is constructed by selecting the $m_i$-th row ($\forall \ m_{i} \in {{\cal M}_0}$) from the DFT matrix $\boldsymbol{F}_{M}$. As the RSs are uniformly inserted into $M_0$ subcarriers, we have $m_{i+1}-m_{i}=m_{i+2}-m_{i+1}, i =1,2,\cdots,M_0-2, \forall m_{i} \in {{\cal M}_0}$. Thus, the autocorrelation of the partial DFT matrix is given by\footnote{The total number of subcarriers is divisible by the number of subcarriers inserted with RSs, i.e., $M/M_0 \in \mathbb{N}$, as specified by the 3GPP standard\cite{3GPPTS38_212}.}
\begin{equation}\label{Eqs0A003}
\boldsymbol{\bar{F}}^H\boldsymbol{\bar{F}} = M_0\left[ {\begin{array}{*{20}{c}}
{{\boldsymbol{I}_{{M_0}}}}& \cdots &{{\boldsymbol{I}_{{M_0}}}}\\
 \vdots & \ddots & \vdots \\
{{\boldsymbol{I}_{{M_0}}}}& \cdots &{{\boldsymbol{I}_{{M_0}}}}
\end{array}} \right] \in \mathbb{C}^{M \times M}.
\end{equation}
Furthermore, note that the last $(M-K)$ rows of the CIR matrix $\boldsymbol{G}$ are zero-padded, i.e.,
\begin{equation}\label{Eqs0A004}
  \boldsymbol{G} = \left[ {\begin{array}{*{20}{c}}
\boldsymbol{\bar{G}}\\
\textbf{0}_{(M-K) \times (N+1)}
\end{array}} \right] \in \mathbb{C}^{M \times (N+1)},
\end{equation}
where $\boldsymbol{\bar{G}} \in \mathbb{C}^{K \times (N+1)}$ denotes the first $K$ rows of the CIR matrix, $\boldsymbol{G}$.
Thus, we have
\begin{equation}\label{Eqs0A00502}
\boldsymbol{\bar {G}}^H\boldsymbol{\bar {G}} = \boldsymbol{G}^H\boldsymbol{G}.
\end{equation}
Since $M_0 \geq K$, based on (\ref{Eqs0A003}), (\ref{Eqs0A004}) and (\ref{Eqs0A00502}), we have
\begin{equation}\label{Eqs0A005}
\begin{split}
\boldsymbol{G}^H\boldsymbol{\bar{F}}^H\boldsymbol{\bar{F}}\boldsymbol{G} = {M_0}\boldsymbol{\bar {G}}^H\boldsymbol{{I}}_{K}\boldsymbol{\bar {G}}= {M_0}\boldsymbol{\bar {G}}^H\boldsymbol{\bar {G}}= {M_0}\boldsymbol{ {G}}^H\boldsymbol{ {G}}.
\end{split}
\end{equation}
By substituting (\ref{Eqs0A005}) into (\ref{Eqs0A002}), we have
\begin{equation}\label{Eqs0A006}
\begin{split}
  \tilde{p}\left(\boldsymbol{v}\right) = \frac{P}{M}{\boldsymbol{v}^H\boldsymbol{G}^H\boldsymbol{G}\boldsymbol{v}} + \sigma^2= {\boldsymbol{v}^H{\boldsymbol{R}}\boldsymbol{v}} + \sigma^2.
\end{split}
\end{equation}
Thus, the proof is completed.

\bibliographystyle{IEEEtran}
\bibliography{mybib}
\end{document}